\documentclass[runningheads]{llncs}

\usepackage[latin1]{inputenc}
\usepackage{amsmath}
\usepackage{amsfonts}
\usepackage{amssymb}
\usepackage{graphicx}
\usepackage{wrapfig}

{\bfseries}{\itshape}

\newcommand{\Parp}[1]{\Pi_p}

\usepackage{graphicx,hyperref}
\usepackage{amsmath}
\usepackage{color}
\usepackage{sgame}
\usepackage[subrefformat=parens,labelformat=parens]{subcaption}

\newcommand{\set}[1]{\left\{{#1}\right\}}
\newcommand{\paren}[1]{\left({#1}\right)}
\newcommand{\card}[1]{\left|{#1}\right|}
\newcommand{\bra}[1]{\left[{#1}\right]}

\let\m=\mu

\usepackage{amsfonts}
\newcommand{\N}{\mathbb{N}}

\newcommand{\eps}{\varepsilon}

\newif\ifshort
\shortfalse  

\ifshort
\renewcommand{\paragraph}[1]{\emph{#1}}
\else 
\fi

\ifshort
\else
\hypersetup{colorlinks=true,linkcolor=blue,citecolor=blue,filecolor=blue,urlcolor=blue}
\fi

\begin{document}
	
	\ifshort
	\title{Independent Lazy  Better-Response Dynamics on Network Games\thanks{A full version of this work is available online at \cite{full}.}} 
	\else 
	\title{Independent Lazy  Better-Response Dynamics on Network Games} 
	\fi

	\author{Paolo Penna\inst{1} \and Laurent Viennot\inst{2}\thanks{Supported by IRIF (CNRS UMR 8243) and Inria project-team GANG.}}
	
	\institute{Department of Computer Science, ETH Zurich, Switzerland, (\texttt{paolo.penna@inf.ethz.ch}) \and  Inria -- Universit\'e Paris Diderot, France, \\ (\texttt{Laurent.Viennot@inria.fr})}

	\maketitle

	\begin{abstract}
		We study an \emph{independent} best-response dynamics on network games in which the nodes (players) decide to revise their strategies independently with some probability. We provide several bounds on the \emph{convergence time} to an \emph{equilibrium} as a function of  this probability,  the degree of the network, and the potential of the underlying games. These dynamics are somewhat more suitable for distributed environments than the classical better- and best-response dynamics where players revise their strategies ``sequentially'', i.e., no two players revise their strategies simultaneously.
	\end{abstract}

\section{Introduction}
Complex and distributed systems are often modeled by means of \emph{game dynamics} in
which the participants (players) act spontaneously, typically striving  to maximize their own
payoff.   
Such selfish behavior often results in a so-called  (pure Nash) \emph{equilibrium} which, roughly speaking, corresponds to the situation in which no player has an incentive to change her current strategy.\footnote{In this work we consider only pure Nash equilibria, which are the equilibria that occur in certain games when each player chooses one strategy out of the available ones. Other equilibrium concepts are also studied, most notably the mixed Nash equilibrium, where each player chooses a probability distribution over the available strategies.} 

Consider the natural scenario in which people  interact on a (social) network and take their decisions based on both their personal interests and also on what their friends decided. Situations of this sort are often modeled by means of \emph{games} that are played locally by the nodes of some graph (see, e.g., \cite{ellison} and \cite[Chapter 19]{EasKle12}).
For example, players may have to choose between  two  alternatives (strategies), and each strategy becomes more valuable if other friends also choose it (perhaps it is easier to agree than to disagree, or it is better to adopt the same technology for working, rather than different ones).  

In many cases, an extremely simple procedure to convergence to an equilibrium is the so-called  \emph{best-response} dynamics in which at each step one player revises her strategy so to maximize her own payoff (and the  others stay put). These dynamics work in more general settings (not only on network games), where convergence to an equilibrium is proven via a potential argument (every move reduces the value of a global function -- called potential).  Games of this nature are called potential games and they are used to model a variety of situations. Interestingly, this argument \emph{fails} as soon as two or more players \emph{move at the same time}. 

In this work we study a natural variant of best-response dynamics in which we relax the requirement that one player at a time moves. That is, now players become active \emph{independently} with some probability and all active players revise their strategy according to the best-response rule (or more generally any better-response rule). This is similar as before but allowing simultaneous moves. Specifically, we study the \emph{convergence time} of these dynamics when players play on a network a ``local'' potential game: (1) each player interacts only with her neighbors, meaning that the strategies of the non-neighbors do not affect the payoff of this player,  and (2) locally the game is a potential game (see Section~\ref{sec:model} for more details and formal definitions). 

Simple examples show that convergence is impossible if two players are always active (move all the time), or that the time to converge can be made arbitrarily long if they become active at almost every step. At the other extreme, if the probability of becoming active is too small, then the dynamics will also take a long time to converge since almost all the time nothing happens. The trade-off is between having sufficiently many active players and, at the same time, not too many neighboring players moving simultaneously.

\subsection{Our Contribution}
We investigate how the convergence time depends on the probabilities of
becoming active and on the degree of the network. This is also motivated by the
search for simple dynamics that the players can easily implement without
global knowledge of the network (namely, they only need to known how many
neighbors they have), nor without having complex reasoning (they still
myopically better-respond).  We first show that for the symmetric
  coordination game,
the convergence time is polynomial whenever the probability of being active is
slightly below the inverse of the maximum degree of the network
(Theorem~\ref{th:upper_simple} and Corollary~\ref{cor:upper_simple}). This
generalizes to \emph{arbitrary} potential games on graphs, where every node
plays a possibly \emph{different} potential game with each of its neighbors,
and the maximum degree is replaced by a \emph{weighted maximum degree} (see
Theorem~\ref{th:upper_gen}). 
These results indeed hold whenever each
active player uses a \emph{better response} (not necessarily the best response).
Finally, we prove a lower bound saying that, in
general, the probabilities of becoming active \emph{must} depend on the degree
for otherwise the convergence time is \emph{exponential} with high probability
(Theorem~\ref{th:lower-bound} and Corollary~\ref{cor:lower-bound}). Note that
this holds also for the simplest scenario of symmetric coordination games.

Our upper bounds can be seen as a probabilistic version of the potential argument (under certain conditions, the potential decreases in expectation \emph{at every step} by some fixed amount).  To the best of our knowledge, this is the first study on the convergence time of these natural variants of best-response dynamics. Prior studies (see next section) either focus on sufficient conditions to guarantee convergence to Nash equilibria, or they consider \emph{noisy} best-response dynamics whose equilibria can be different from best-response. 

We note that the general upper bound necessarily depends on the 
maximum value of the potential, as these games include \emph{max-cut} games which are PLS-complete \cite{SchYan91}: for such games, no centralized algorithm for computing a Nash equilibrium in time polynomial in the number of players is known, and these games are hard precisely when the potential can assume arbitrarily large values. Obviously, one cannot hope that simple distributed dynamics do much better than the best centralized procedure.

\subsection{Related Work}
Several works study convergence to Nash equilibria for simple variants of best-response dynamics. A first line of research concerns the ability to converge to a Nash equilibrium when the strict schedule of the moves of the players (one player at a time) is relaxed \cite{CauDurGauTou14}; they proved that any ``separable'' schedule guarantees convergence to a Nash equilibrium.  Other works study the convergence time of specific dynamics with limited simultaneous moves:  
\ifshort
\cite{FotKapSpi10} introduce a ``local'' coordination mechanism for congestion games (which are equivalent to potential games \cite{MonSha96}), while \cite{FanMosSko12} shows that with   limited simultaneous moves the dynamics reaches quickly a state whose cost is not too far from the worst Nash equilibrium \cite{FanMosSko12}; Fast convergence can be achieved in certain linear congestion games if \emph{approximate} equilibria are considered \cite{ChiSin15}.
\else
\cite{FotKapSpi10} introduce a ``local'' coordination mechanism for congestion games (which are equivalent to potential games \cite{MonSha96}), while \cite{FanMosSko12} shows that with   limited simultaneous moves the dynamics reaches quickly a state whose cost is not too far from the worst Nash equilibrium \cite{FanMosSko12}; Note that the reached state need not be a Nash equilibrium, and the required condition is that every $T$ consecutive best responses, each player has moved at least once and at most $\beta$ times.  Fast convergence can be achieved in certain linear congestion games if \emph{approximate} equilibria are considered \cite{ChiSin15}, that is, players keep changing their strategies as long as a significant improvement is possible.  
\fi

Another well-studied variant of best-response dynamics is that of \emph{noisy}
or \emph{logit (response)} dynamics \cite{Blu93,Blu98,AloNet10}, where players'
responses is \emph{probabilistic} and determined by a noise parameter (as the
noise tends to zero, players select almost surely best-responses, while for
high noise they respond at random).  These dynamics turn out to behave
differently from ``deterministic'' best-response in many aspects. In the
original logit dynamics by \cite{Blu93,Blu98}, where one randomly chosen player
moves at a time, they essentially rest on a subset of potential
minimizers. When the players' schedule is relaxed, this property is lost and
additional conditions on the game are required
\cite{AloNet10,CauDurGauTou14,AloNet15,FerPen15,Pen15}.  
Our \emph{independent} better-response dynamics  can be seen as an analog of the independent dynamics of \cite{AloNet10} for logit response.

Potential games on graphs (a proper subclass of potential games) are well-studied because of their  many applications. In physics,  ferro-magnetic systems are modeled as \emph{noisy} best-response dynamics on lattice graphs in which every player (node) plays a coordination game with each neighbor (see, for example, \cite{Mar99} and
Chapter 15 of \cite{LevPerWilLee09}). The version in which the coordination game is asymmetric (i.e., coordinating on one strategy is more profitable than another) is used to model the diffusion  of new technologies \cite{MonSab10,KreYou14} and opinions \cite{FerGolVen12} in social networks. 
Finally, potential games on graphs (every node plays some potential game with each neighbor)
characterize the class of potential games for which the equilibria of \emph{noisy} best-response dynamics with \emph{all} players updating simultaneously can be ``easily'' computed \cite{AulFerPasPenPer15}. 
The convergence time of best-response dynamics for games on graphs is studied in \cite{DyeMoh11,FerGolVen12}: Among other results, \cite{DyeMoh11} showed that a polynomial number of steps are sufficient when the same game is played on all edges and the number of strategies is constant. Analogous results are proven for finite opinion games in \cite{FerGolVen12}. Finally, \cite{BabTam14} characterize the class of potential games which are also \emph{graphical} games \cite{Kea01}, where the potential can be decomposed into the sum of potentials of ``maximal'' cliques of an underlying graph. Graphical games have been studied in several works (see, e.g., \cite{DasPapICALP09,YanDasSODA11,PilShaSODA14,BaiPilEC18}). The class of \emph{local interaction} potential games \cite{AulFerPasPenPer15} is the restriction in which the potential can be decomposed into pairwise (edge) potential games. In this work we deal precisely with this class of games. 
Since this class includes the so-called max-cut games, which are known to be PLS-complete \cite{SchYan91}, it is considered unlikely that an equilibrium can be computed efficiently, even by a centralized procedure.

Our dynamics are similar to the $\alpha$-synchronous dynamics  in  cellular automata~\cite{FatesRST06}. 
\ifshort
In particular, the case of symmetric coordination game corresponds to majority rule on general graphs~\cite{RouquierRT11} (where each cellular automaton tries
to switch to the majority state of its neighbors, and stays put in case of ties). The present work can be seen as a first study of
$\alpha$-synchronous dynamics on general graphs for the rules that follow from
best-response to some potential games with neighbors. 
\else
This latter model encompasses
ours as long as the response strategy can be implemented through an
automata. This is the case for example with symmetric coordination where
best-response corresponds to majority rule (where each cellular automaton tries
to switch to the majority state of its neighbors).  However, most of the
literature around cellular automata considers a regular lattice topology and do
not cover the case of general graphs we consider here.  A notable exception is
the study of minority rule on general graphs~\cite{RouquierRT11} with similar
dynamics as ours.  The present work can be seen as a first study of
$\alpha$-synchronous dynamics on general graphs for the rules that follow from
best-response to some potential games with neighbors. Our results apply in
particular to majority and minority rules (where no action is taken in case of
equality, i.e, when the two most frequent states are equally represented among
neighbors).
\fi

\section{Model (Local Interaction Potential Games)}\label{sec:model}

Intuitively speaking we consider a network (graph) where each node is a player
who repeatedly plays with her neighbors.  We assume that a two-player potential
game (defined below) is associated to each edge of the graph. Each player must play the same
strategy on all the games associated to its incident edges, and her payoff is the
sum of the payoffs obtained in each of these games. 
We also assume finite strategies, i.e. each player chooses her strategy within a finite set.

\paragraph{Symmetric Coordination Game.} 
One of the simplest (potential) games is the \emph{symmetric coordination game} where each player
chooses color $B$ or $W$ (for black or white) and her payoff is $1$ if players agree on their strategies, and $0$ otherwise (see Figure~\ref{fig:symmetric-coordination} where the two numbers are the payoff for the row and the column player, respectively).

 \begin{figure}[tb]
 	\begin{subfigure}[b]{0.24\textwidth}
 		\centering
 		\begin{game}{2}{2}[][][]
 			& $B$ & $W$ \\
 			$B$ & $1,1$    & $0,0$\\
 			$W$ & $0,0$    & $1,1$
 		\end{game}
 		\ifshort
 		\vspace{-.5cm}
 		\fi
 		\caption{Symmetric \\ Coordination Game.}
 		\label{fig:symmetric-coordination}
 	\end{subfigure}
 	\hfill
 	\begin{subfigure}[b]{0.24\textwidth}
 		\centering
 		\begin{game}{2}{2}[][][]
 		& $B$ & $W$ \\
 	$B$ & $2,1$    & $0,0$\\
 	$W$ & $0,0$    & $1,2$
	 	\end{game}
	 	\ifshort
	 	\vspace{-.5cm}
	 	\fi
 		\caption{Another \\ Coordination Game.}
 		\label{fig:coordination}
 	\end{subfigure}
 	\hfill
 	\begin{subfigure}[b]{0.24\textwidth}
 		\centering
 		\begin{game}{2}{2}[][][]
 							& $B$ & $W$ \\
 							$B$ & $-2$    & $-1$\\
 							$W$ & $0$    & $-2$
 						\end{game}
 						\ifshort
 						\vspace{-.5cm}
 						\fi
 		\caption{Potential \\  for Game \subref{fig:coordination}.}
 		\label{fig:potential}
 	\end{subfigure}
 	\hfill
 	\begin{subfigure}[b]{0.24\textwidth}
 		\centering
 		\begin{game}{2}{2}[][][]
 							& $B$ & $W$ \\
 							$B$ & $2,2$    & $1,1$\\
 							$W$ & $0,0$    & $2,2$
 						\end{game}
 						\ifshort
 						\vspace{-.5cm}
 						\fi
 		\caption{Game equivalent to Game \subref{fig:coordination}.}
 		\label{fig:equiv}
 	\end{subfigure}
 	\ifshort
 	\caption{Examples of two-player games and potential function.}
 	\else 
 	\caption{Examples of (two-players) potential games. The game  \subref{fig:coordination} has potential function \subref{fig:potential}. Game \subref{fig:equiv} has same potential function \subref{fig:potential} (signs changed for payoff), it is equivalent to Game \subref{fig:coordination} (both games have same dynamics).}
 	\fi
 \end{figure}

\paragraph{General Potential Games.} 
In a general game, we have $n$ players, and each of them can choose one color (strategy) and the combination $c=(c_1,\ldots,c_n)$ of all colors gives to each player $u$ some payoff $PAY_u(c)$. 
In a \emph{potential game}, when the change in the payoff of any player improves by some amount, some global function $P$ called the \emph{potential} will be decreased by the same amount: For any player $u$ and any two configurations $c$ and $c'$ which differ only in $u$'s strategy, it holds that 
\begin{equation}
\label{eq:potential-games}
PAY_u(c') - PAY_u(c) = P(c) - P(c') \enspace .
\end{equation}

A configuration $c$ is a \emph{(pure Nash) equilibrium} if no player $u$ can improve her payoff, that is, the quantity above is negative or zero for all $c'=(c_1,\ldots,c_u',\ldots,c_n)$. Conversely, $c$ is not an equilibrium if there is a player $u$ who can improve her payoff ($PAY_u(c') - PAY_u(c)>0$) in which case $c_u'$ is called a \emph{better response} (to strategies $c$). A \emph{best response} is a better response maximizing this improvement, over the possible strategies of the player. 
Potential games possess the following nice feature: A configuration $c$ is an equilibrium if and only if no player can improve the potential function by changing her current strategy.  In a general (two-player) potential game the payoff of the players is not the same, and the potential function is therefore not symmetric (see the example in Figure~\ref{fig:potential}). 

\paragraph{Local interaction potential games \cite{AulFerPasPenPer15}.}
In a local interaction potential game the potential function can be decomposed into the sum of two-player potential games, one for each edge of the network $G$:
\begin{equation}\label{eq:local-interacion-potential}
P(c) = \sum_{uv \in E(G)} P_{uv}(c_u,c_v)\enspace .
\end{equation}
No edge exists if the strategies of the two players do not affect each others' payoff (the corresponding potential is constant and can be ignored). 
This definition captures the following natural class of games on networks: Each edge corresponds to some potential game, and the payoff of a player is the sum of the payoffs of the games with the neighbors. Note that a player chooses one strategy to be played on all these games.

\paragraph{(Independent) Better-Response Dynamics.} A simple procedure for computing an equilibrium consists of repeatedly selecting \emph{one} player who is currently not playing a best response and let her play a better or best response. Every step reduces the potential by a finite amount, and therefore this procedure terminates into an equilibrium in $O(M)$ time steps, where $M$ is the maximum value for the potential (w.l.o.g., we assume that the potential is always non-negative and takes integer values\footnote{As we assume that strategy sets are finite, the potential function is defined by a finite set of values. Rescaling the potential function so that different values are at least 1 apart, and then truncating the values to integers allows to obtain an equivalent game (with same dynamics). Additionally shifting the values allows to obtain a non-negative potential function for that game.}). Here we consider the variant in which, at each time step, each player becomes \emph{independently} active according to some probability, and those who can improve their payoff change strategy accordingly: 

\ifshort
\begin{definition}
	In independent better-response dynamics, at each time  step $t $ players do the following: 
	
	\noindent
	-- Each player (node) $u$ becomes active with some probability $p_u^t$ which can change over time 
	(the case in which it is constant over time is a special case of this one).
	
	\noindent
	-- Every active player (node) revises her strategy according to a better (or best) response rule. If the current strategy is already a best response, then no change is made. 
\end{definition}
\else
\begin{definition}
	In independent better-response dynamics, at each time  step $t $ players do the following: \begin{itemize}
		\item Each player (node) $u$ becomes active with some probability $p_u^t$ which can change over time 
(the case in which it is constant over time is a special case of this one).
		\item Every active player (node) revises her strategy according to a better (or best) response rule. If the current strategy is already a best response, then no change is made. 
	\end{itemize}
\end{definition}
\fi

Note that all players that are active at a certain time step may change their strategies \emph{simultaneously}. So, for example, it may happen that on the symmetric coordination game in Figure~\ref{fig:symmetric-coordination} the two players move from state $BW$ to state $WB$ and back if they are both active all the time.

\ifshort
\paragraph{Generic upper bound.}
\else
\subsubsection{Generic upper bound}
\fi

To show that dynamics converge quickly, we show that the potential decreases in expectation at every step. To this end, we consider the probability space
of all possible evolutions of the dynamics. A configuration $c$ at a given
time $t$ is given by the colors chosen by players at the previous time step
(strategy profile) and by the values $p_u^t$ used by users for randomly deciding to be active at time $t$.
The universe $\Omega$ is then defined
as the set of
all infinite sequences $c^0,c^1,\ldots$ of configurations.

\begin{definition}[$\delta$-improving dynamics]
	Dynamics are $\delta$-improving for a given (local interaction) potential game if in expectation the potential decreases by at least $\delta$ during each time step, unless the current configuration is an equilibrium. That is, for any 
	configuration $c$ which is not an equilibrium, and any event 
	$F_c^t=\{c^0,c^1,\ldots\in\Omega \mid c^t = c\}$
	where configuration $c$ is reached at time $t$, we have
	\[E[P^{t+1}-P^t \mid F_c^t] \le -\delta\]
	where $P^t$ denotes the potential at time $t$.
\end{definition}

\ifshort
Standard Martingale arguments imply the following  (see \cite{full} for details):
\else
The proof of the following theorem is based on standard Martingale arguments  and it is given in Appendix~\ref{sec:proof:th_upper_any_game} for completeness.
\fi

\begin{theorem}
	\label{th:upper_any_game}
	The expected convergence time of any $\delta$-improving dynamics is $O\left(\frac{M_0}{\delta}\right)$ where $M_0$ is the expected potential of the game at time 0.
\end{theorem}

\section{Networks With Symmetric Coordination Games} 
\label{sec:upper_simple}

We first consider the scenario in which every edge of the network is the  symmetric coordination game in Figure~\ref{fig:symmetric-coordination}.  
The nodes of a graph $G$ (players)
can choose between two colors $B$ and $W$ and are rewarded according to the 
number of neighbors with same color. We are  thus
considering the dynamics in which nodes attempt to
choose the \emph{majority color of their neighbors} and every active node changes its color if more than half of its neighbors has the different color.

In order to analyze the convergence time of these
dynamics, we shall relate the probabilities of being active to the number of
neighbors having a different color. We say that $u$ is \emph{unstable} at time
$t$ if more than half of the neighbors has the other color, that is,
\[
dc_u^t>\frac{1}{2}\delta_u
\] 
where $\delta_u$ is the degree of $u$ and $dc_u^t$ is the number of neighbors of $u$ that have a color different from the color of $u$ at time $t$.
By definition, the dynamics converge if no node is unstable. 
Note that we have $dc_u^t\le \delta_u\le \Delta_u\le \Delta$ where
$\Delta=\max_{u\in V(G)}\delta_u$ is the maximum degree of the graph,
and $\Delta_u=\max_{uv\in E(G)}\delta_v$ is the local maximum degree in
the neighborhood of $u$.

For the case of symmetric coordination games, the potential function of a configuration is the number of edges whose endpoints have different colors:
An edge $uv$ is said to be \emph{conflicting} in configuration $c$
if $u$ and $v$ have different colors.    Therefore the potential is at most the number $m$ of edges.

\begin{theorem}
	\label{th:upper_simple}
	Fix some real values $p,q\in(0,1)$.
	If we have $p_u^t \in [\frac{p}{\Delta}, \frac{q}{\Delta_u}]$ for all $u,t$ in a symmetric coordination game, then
	the expected convergence time is $O\left(\frac{\Delta m_0}{p(1-q)}\right)$ where $m_0$ is the initial number of conflicting edges,
$\Delta$ is the maximum degree, and
$\Delta_u$ is the maximum degree in the neighborhood of $u$.
\end{theorem}

As an immediate corollary, we have the following result for the case in which all nodes are
active with the \emph{same} probability probability $p$.

\begin{corollary}
	\label{cor:upper_simple}
	If all unstable nodes are active with probability $p<\frac{1-\eps}{\Delta}$
	for $\eps > 0$, then
	the dynamics converge to a stable state in $O(\frac{m_0}{p\eps})$ expected time.
\end{corollary}

Theorem~\ref{th:upper_simple} derives from the following lemma and Theorem~\ref{th:upper_any_game}.

\begin{lemma}
	\label{lem:upper_simple}
	Any dynamics satisfying the hypothesis of Theorem~\ref{th:upper_simple} are $\delta$-improving for $\delta=p(1 - q)/\Delta$. 
\end{lemma}

\begin{proof} 
  Consider the event $F_c^t$ where a configuration $c$ is reached at time $t$.
  Let $C^t$ denote the
	number of conflicting edges in $c$, 
	and $U^t$ be the set of unstable nodes at time $t$ respectively. Recall that the number of conflicting edges is equal to the potential, that is, $P^t=C^t$.  
        We now express $E[C^{t+1}-C^{t} \mid F_c^t]$ as a function
        of the values $\set{p_u^t \mid u\in V(G)}$ associated to $c$.
	
	For that purpose, we first analyze the probability that any given edge
	of $c$ is conflicting after the random choices made at time $t$.
	We distinguish the following types of edges.
	Let $S_1$ (resp. $S_2$) denote the set of edges in $c$ with the same color 
	and one unstable extremity (resp. two). Similarly, let $C_1$ (resp. $C_2$)
	denote the set of edges in $c$ with conflicting colors and one unstable
	extremity (resp. two). Note that $C^t = |C_1| + |C_2|$.
	A conflicting edge $uv$ will become non-conflicting if only one extremity
	changes its color. Similarly, a non-conflicting edge $uv$
	will become conflicting if only one extremity changes its color. 
	Due to independence of choices, this happens in both cases
	with probability $p_{uv}^t=p_u^t(1-p_v^t) +
	(1-p_u^t)p_v^t$ if both $u$ and $v$ are unstable, and with probability $p_u^t$
	if $u$ is unstable and $v$ is not. 
	By linearity of expectation, we then obtain:
	\begin{equation}
	\label{eq:deltapot}
	E[C^{t+1}-C^t\mid F_c^t] = 
	\sum_{uv\in S_1}p_u^t
	+ \sum_{uv\in S_2}p_{uv}^t
	- \sum_{uv\in C_1}p_u^t
	- \sum_{uv\in C_2}p_{uv}^t \enspace .
	\end{equation}
	(When we note $uv\in C_1$ (resp. $uv\in S_1$), we assume that $u$
	is unstable and $v$ is not.)
	By definition, each unstable node $u$ sees  more conflicting edges than non-conflicting ones, thus implying
	\ifshort
		$
		1 + \sum_{v|uv\in S_1}1
		+ \sum_{v|uv\in S_2} 1
		\le
		\sum_{v|uv\in C_1}1
		+ \sum_{v|uv\in C_2} 1.
		$
		\else
	\fi
	By multiplying by $p_u^t$ and then summing over all unstable nodes, we  obtain:
	\begin{equation}
	\label{eq:active}
	\sum_{u\in U^t}p_u^t + \sum_{uv\in S_1}p_u^t
	+ \sum_{uv\in S_2}(p_u^t+p_v^t)      
	\le
	\sum_{uv\in C_1}p_u^t
	+ \sum_{uv\in C_2}(p_u^t+p_v^t) \enspace .
	\end{equation}
	As $p_{uv}^t=p_u^t+p_v^t-2p_u^tp_v^t$, we deduce from 
	\eqref{eq:deltapot} and \eqref{eq:active}:
	
	\begin{equation}
	\label{eq:deltapot2}
	E[C^{t+1}-C^t \mid F_c^t] \le
	\sum_{uv\in C_2} 2p_u^tp_v^t - \sum_{u\in U^t} p_u^t \enspace .
	\end{equation}
	Since every edge $uv\in C_2$ has both endpoints in $U^t$, we can rewrite \eqref{eq:deltapot2} as 
	
	\[
	E[C^{t+1}-C^t \mid F_c^t]
	\le \sum_{u\in U^t} p_u^t \Big(-1 + \sum_{v | uv \in C_2} p_v^t\Big) \enspace .
	\]
Using $p_v^t \le \frac{q}{\Delta_v} \le \frac{q}{\delta_u}$
and $p_u^t \ge \frac{p}{\Delta}$, we obtain the following inequality:
	$
	E[C^{t+1}-C^t \mid F_c^t] 
	\le \sum_{u\in U^t} \frac{p}{\Delta} (-1 + q) = - p(1-q) \frac{|U^t|}{\Delta}.
	$
	This completes the proof.
\qed\end{proof}

\ifshort
\paragraph{Adaptive Probabilities.}
\else
\subsubsection{Adaptive probabilities}
\fi
The upper bound of Theorem~\ref{th:upper_simple} can be improved if nodes are aware of the number of neighbors that are willing to change strategy (unstable) and then set accordingly the probability of changing too. More precisely, one can think of active nodes announcing to their neighbors that they are unstable and that they would like to switch to the other color, before actually doing so. Then, each unstable node will switch with a probability inversely proportional to the number of unstable neighbors. The following theorem shows that this yields an improved upper bound on the convergence time. 

\begin{theorem}
	\label{th:adaptive}
	Fix some real values $p,q\in \left(0,\frac{1}{2}\right)$. If  we have 
	$p_u^t\in [\frac{p}{d_u^t+1}, \frac{q}{d_u^t+1}]$ for all $u,t$ in a symmetric coordination game, 
	where $d_u^t$ is the number of conflicting unstable neighbors of $u$, then the
	expected convergence time is $O\left(\frac{m_0}{p(1-2q)}\right)$ where $m_0$ is the initial number of conflicting edges.
\end{theorem}

To prove this theorem we adapt  the proof of Lemma~\ref{lem:upper_simple} and show that these dynamics are $\delta$-improving for $\delta=p(1-2q)$ 
\ifshort
(see \cite{full} for details).
\else
(see Appendix~\ref{app:proof:th:adaptive} for details).
\fi

\ifshort
\medskip
\paragraph{Fully Local Dynamics.}
\else
\subsubsection{Fully local dynamics}
\fi 
Theorem~\ref{th:upper_simple} requires that each node is aware of a bound on
the maximum degree, or the local maximum degree in her neighborhood for setting
$p_u^t$. Theorem~\ref{th:adaptive} requires knowledge of the number of
conflicting unstable neighbors at each time step. We next consider dynamics that are fully local as each node $u$ can set the probabilities  $p_u^t$ by only looking at its own degree. 

\begin{theorem}
	\label{th:degree}
	Fix some real values $p,q\in \left(0,\frac{1}{2}\right)$. If  we have 
	$p_u^t\in [\frac{p}{\delta_u}, \frac{q}{\delta_u}]$ for all $u,t$ in a symmetric coordination game, 
	where $\delta_u$ is the degree of $u$, then the
	expected convergence time is $O\left(\frac{\Delta m_0}{p(1-2q)}\right)$ where $m_0$ is the initial number of conflicting edges.
\end{theorem}

\ifshort
The proof of this theorem is similar to that of Theorems~\ref{th:upper_simple} and \ref{th:adaptive}  (see \cite{full}).
\else
The proof of this theorem is similar to that of Theorems~\ref{th:upper_simple} and \ref{th:adaptive}  (see Appendix~\ref{app:proof:th:degree}).
\fi

\ifshort
\smallskip
\paragraph{Tightness of the Results.}
Consider the following network  composed of  a clique and $r/2+1$ paths, for even $r$ (see figure below). Each node in a path is connected to all nodes to the right and to the left path (or clique for the first path) as feature by demi-edges with degree indications w.r.t. the previous and the next part of the construction. Below each part, we indicate the number of nodes in the part. 
\begin{wrapfigure}{r}{0.5\textwidth}
	\centering
	\vspace{-.5cm}
	\includegraphics[scale=.35]{adaptive-better_third_larger_fonts}
	\vspace{-.5cm}
\end{wrapfigure}
Intuitively, the construction is such that the process proceeds from left to right, where nodes in certain path become unstable only after all nodes in the previous path became black; moreover, inside each path the process is also sequential, i.e.,  the path becomes black from extremities to center.
\else
\subsubsection{Tightness of the results}
We can use the example in Figure~\ref{fig:adpative_dynamics_better_second} to show that the analysis of Theorems~\ref{th:upper_simple}, \ref{th:adaptive}, and \ref{th:degree} is tight, and adaptive dynamics are provably faster than non-adaptive ones. 
	\begin{figure}[tb]
		\centering
		\includegraphics[scale=.5]{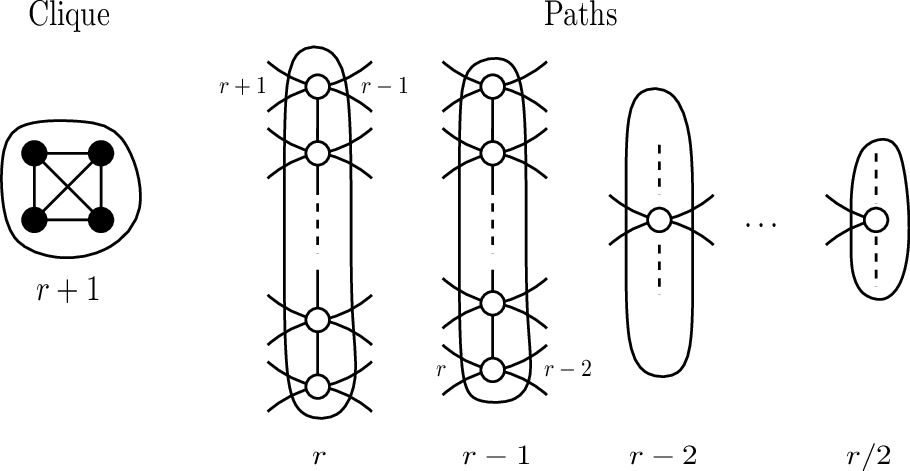}
		\caption{Example of network with faster adaptive dynamics than non-adaptive ones. The network is composed of $r/2+2$ parts for even $r$: a clique and $r/2+1$ paths. Below each part, we indicate the number of nodes in the part. Each node in a path is connected to all nodes to the right and to the left path (or clique for the first path) as feature by demi-edges with degree indications w.r.t. the previous and the next part of the construction.}
		\label{fig:adpative_dynamics_better_second}
	\end{figure}
	Intuitively, the construction is such that the following holds:
		\begin{itemize}
			\item Initially only the extreme nodes of the leftmost path are unstable. After an unstable node become black,  the next neighbor in the path (and only that one) becomes unstable.
			\item The process proceeds from left to right, and nodes in certain path become unstable only after all nodes in the previous path became black (except for the last path with $r/2$ nodes). 
			\item Inside the path containing the only unstable nodes, the process is sequential: the path becomes black from extremities to center. At most two nodes are unstable: the extremities of the sub-path of white nodes. 
	 meaning that after one node changes color, its neighbor (and only that) becomes unstable. 
		\end{itemize}
\fi 
These observations imply that any dynamics in which nodes become active with probability $p\simeq \alpha$, require $\Omega(r^2/\alpha)=\Omega(n/\alpha)$ steps. 
	
	Since every node has degree $\Theta(r)=\Theta(\sqrt{n})=\Theta(\Delta)$, and the initial configuration has $m_0=\Theta(r^2)=\Theta(n)$ conflicting edges (those between the clique and the first path),  non-adaptive dynamics take 
	$\Theta(\Delta m_0) = \Theta(n^{3/2})$ time steps. On the contrary, adaptive dynamics take $\Theta(m_0)=\Theta(n)$ steps since the number $d_u^t$ of unstable conflicting neighbors of each node $u$ is at most $1$. Therefore, the analysis of Theorems~\ref{th:upper_simple}, \ref{th:adaptive}, and \ref{th:degree} is tight. Moreover, the adaptive dynamics are provably faster than non-adaptive ones.

\section{An Exponential Lower Bound when the Degree is Unbounded} 

In this section we prove a lower bound for the case of symmetric coordination game on each edge and dynamics with constant probabilities, that is, the case in which every node becomes active with some probability $p$ which does not depend on the graph nor on the time, and which is the same over all nodes.  

\begin{theorem}\label{th:lower-bound}
	For every $p>0$, there are starting configurations of the complete bipartite graph where the expected number of steps to converge to an equilibrium  is exponential in the number of nodes. 
\end{theorem}

\paragraph{Proof Idea.}
Consider the \emph{continuous} version of the problem in which, instead of a bipartite graph with $n$ nodes on each side, we imagine $L$ and $R$ being two continuous intervals (see figure below). Start from a symmetric configuration in which a fraction $\alpha>1/2$ of the players in $L$ has color $W$ and the same fraction in $R$ has the other color $B$. 
Suppose that 
\ifshort
$
\alpha = \frac{1}{2-p}.
$
\else
\[
\alpha = \frac{1}{2-p}.
\]
\fi 
Then after one step the system reaches the symmetric configuration, that is, a fraction $\alpha$ of nodes in $L$ has color $B$ and the same fraction in $R$ has color $W$. Indeed, the fraction $\beta$ of players with color $B$ in $L$ after one step is precisely
\ifshort
$
\beta = 1 - \alpha + p \cdot\alpha = \frac{1-p}{2-p} + \frac{p}{2-p} =\alpha.
$
\else 
\[
\beta = 1 - \alpha + p \cdot\alpha = \frac{1-p}{2-p} + \frac{p}{2-p} =\alpha \ .
\]
\fi 

\ifshort
\begin{wrapfigure}{r}{0.5\textwidth}
	\centering
	\includegraphics[scale=.35]{bipartite-idea_cameraready}
	\vspace{-.95cm}
\end{wrapfigure}
\else 
\begin{figure}
	\centering
	\includegraphics[scale=.5]{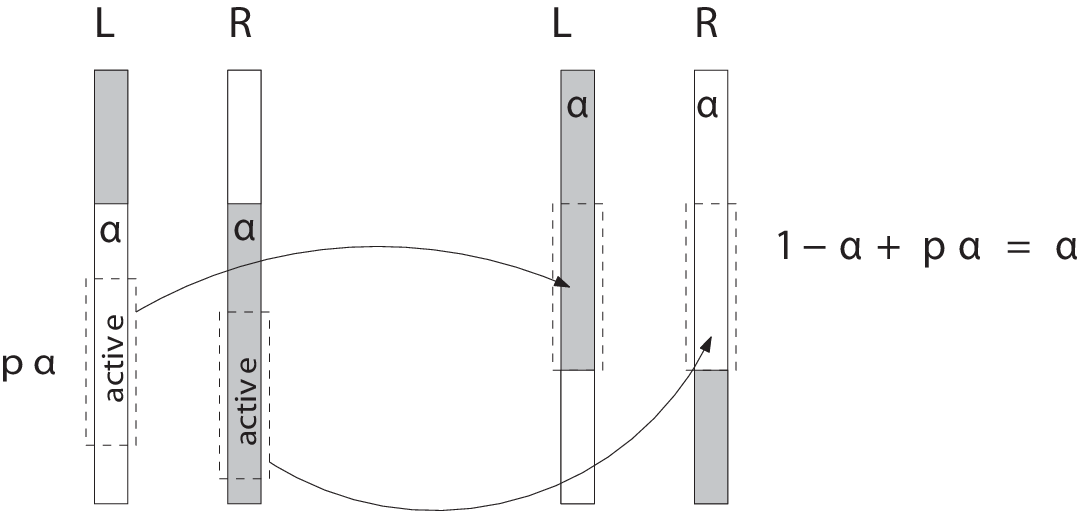}
	\caption{The idea of the proof of Theorem~\ref{th:lower-bound} (active nodes change color leading to a symmetric configuration for suitable $\alpha$). All the nodes on the left ($L$) are adjacent to all nodes on the right ($R$).}
	\label{fig:bipartite-idea}
\end{figure}
\fi 
We next prove the theorem via Chernoff bounds. For $\epsilon =p/3$ consider the interval
$
around(\alpha) := [(1-\epsilon)\alpha, (1+\epsilon)\alpha],
$
and let 
$CYCLE(t)$ be the following event:
\begin{quote}
	$CYCLE(t):=\{$At time $t$ a fraction $\alpha_L\in around(\alpha)$ of the nodes in $L$ has some color $c$, and a fraction $\alpha_R\in around(\alpha)$ of the nodes in $R$ has the other color $\overline{c}$ (where $\overline{B}=W$
and $\overline{W}=B$).$\}$
\end{quote}
We say that the configuration is \emph{balanced} at time $t$ when $CYCLE(t)$ holds.
Since $\epsilon<p/2$ we have  $(1-\epsilon)\alpha>1/2$, and thus the best response of every (active) node in a balanced configuration is to switch color (since both $\alpha_Ln$ and $\alpha_R n$ are strictly larger than $n/2$). 
\ifshort
Chernoff bounds guarantee that with high probability enough many nodes will be activated and therefore will switch to obtain a symmetric balanced configuration  (see \cite{full} for  proof of next lemma):
\else
Chernoff bounds guarantee that with high probability enough many nodes will be activated and therefore will switch to obtain a symmetric balanced configuration.  Formally, the following lemma holds (see Appendix~\ref{app:proof:le:cycle} for the proof).
\fi

\begin{lemma}\label{le:cycle}
	For any $t$, it holds that
	$
	P[CYCLE(t+1)| \ CYCLE(t)] \geq  1 - 4\exp\left(-\frac{\delta^2}{3}\mu\right),
	$
	where $\delta=\frac{\epsilon}{1+\epsilon}$  
        and $\mu = p(1+\epsilon)\alpha n$ with $\epsilon=p/3$.
\end{lemma}

\ifshort
The above lemma implies that, starting from a balanced configuration, the probability of reaching an equilibrium in $t$ steps is at least $(1-q)^{t-1}$ where $q=4\exp(-\frac{\delta^2}{3}\mu)$. The expected time to converge is thus at least $1/q^2$ which proves the theorem. 
\else
\begin{proof}[of Theorem~\ref{th:lower-bound}]
	Consider any starting configuration which is balanced, that is, $CYCLE(0)$ holds.
	By Lemma~\ref{le:cycle}, the probability that the event
	$CYCLE(1)\wedge\cdots\wedge CYCLE(t)$ holds is at least $(1-q)^t$
	where $q=4\exp\left(-\frac{\delta^2}{3}\mu\right)$. 
Since $CYCLE(t)$ implies that at least one node is unstable, the probability that we reach an equilibrium in $t$ steps is thus at least $(1-q)^{t-1}$. The expected time for convergence is thus at least $\sum_{t\ge 0}t(1-q)^{t-1}=1/q^2$.
\qed\end{proof}	
\fi

Simple calculations lead to the following result
\ifshort
(see \cite{full} for details):
\else
(see Appendix~\ref{app:proof::lower-bound} for details).
\fi

\begin{corollary}
  \label{cor:lower-bound}
	Starting from any balanced configuration, the expected number of steps to converge to an equilibrium in the complete bipartite graph is $e^{\Omega(n^{1-3c})}$, as long as $p\geq 1/n^c$ with $0 \leq c<1/3$.
\end{corollary}

\section{General Local Interaction  Potential Games} 

In this section we extend the upper bound of Theorem~\ref{th:upper_simple} to general local interaction potential games:  each edge $uv$ of $G$ is associated with a (two-player) potential
game with potential $P_{uv}$. Without loss of generality, we assume that the potential $P_{uv}$ takes integer  non-negative values.  
The upper bound is given in terms of the following quantity: 
\begin{equation}
\Delta_P := \max_u \sum_{v\in N(u)} \Delta_{P_{uv}} \enspace , \label{eq:max-potential}
\end{equation}
where $\Delta_{P_{uv}}$ denotes the maximum value of
$P_{uv}$. Note that for symmetric coordination games, $\Delta_P$ is simply the maximum degree $\Delta$ of the graph.

\begin{theorem}
	\label{th:upper_gen}
	For any $p,q\in(0,1/2)$,
	if we have $p_u^t \in [\frac{p}{\Delta_P}, \frac{q}{\Delta_P}]$ 
	for all $u$ and $t$ and for $\Delta_P$
	defined as in \eqref{eq:max-potential} in a general local interaction potential game, then
	the expected convergence time is $O\left(\frac{n\Delta_P^2}{p(1-2q)}\right)$.
\end{theorem}

\ifshort
\else
We shall prove later that $n\Delta_P$ is an upper bound on the global potential. Therefore to prove the theorem it suffice to show that in expectation the global potential improves as follows:

\begin{lemma}
	\label{lem:upper_gen}
	Any dynamics satisfying the hypothesis of Theorem~\ref{th:upper_gen} is $\delta$-improving for $\delta=p(1-q)/\Delta_P$.
\end{lemma}

\ifshort
\else
\begin{proof} 
  Consider the event $F_c^t$ where a configuration $c$ is reached at time $t$.
	Let  $P^t:=P(c^t)=\sum_{uv\in E(G)}P_{uv}(c_u^t,c_v^t)$ denote the global
	potential of the game in configuration $c$.
	A node is unstable when it can improve its total payoff 
	by changing $c_u^t$ to some $c'_u$, which corresponds to an equal improvement on the global potential of the game according to \eqref{eq:potential-games}.
	We then let $$\m_{uv}^u=-(P_{uv}(c'_u,c_v^t)-P_{uv}(c_u^t,c_v^t))$$ denote the
	improvement obtained on the game with a neighbor $v$, and the sum $\sum_{v\in N(u)}\m_{uv}^u\ge 1$
	being the total improvement.
	Node $u$ then changes for strategy $c'_u$ with probability $p_u^t$.
	
	Let $C_1$ denote the edges $uv$ such that $u$ is unstable at time $t$
	and $v$ is not. Let $C_2$ denote the edges $uv$ such that both $u$ and $v$
	are unstable. The set of unstable nodes at time $t$ is denoted by $U^t$.
	By linearity of expectation, we have
	\begin{equation}
	\label{eq:deltapot2:gen}
	E[P^{t+1}-P^t\mid F_c^t] =
	-\sum_{uv\in C_1} p_u^t\m_{uv}^u
	- \sum_{uv\in C_2} (p_u^t(1-p_v^t)\m_{uv}^u + p_v^t(1-p_u^t)\m_{uv}^v
	- p_u^tp_v^t \nu_{uv}^{uv})
	\end{equation}
	where $\nu_{uv}^{uv} = P_{uv}(c'_u,c'_v) - P_{uv}(c^t_u,c^t_v)$ denotes the variation in the
	potential of edge $uv$ when both $u$ and $v$ change their colors.
	
	As each unstable node expects a gain of at least 1 (since the potential takes integer values),  we have
	\begin{equation}
	\label{eq:active2}
	\sum_{uv\in C_1} p_u^t\m_{uv}^u + \sum_{uv\in C_2} (p_u^t\m_{uv}^u + p_v^t\m_{uv}^v)
	\ge \sum_{u\in U^t} p_u^t  \enspace .
	\end{equation}
	By combining \eqref{eq:deltapot2:gen} and \eqref{eq:active2} we obtain
	\[
	E[P^{t+1}-P^t\mid F_c^t] \le
	-\sum_{u\in U^t} p_u^t
	+ \sum_{uv\in C_2} p_u^tp_v^t (\m_{uv}^{u} + \m_{uv}^{v} + \nu_{uv}^{uv})  \enspace .
	\]
	The last factor in the second sum simplifies to
	\[
	\m_{uv}^{u} + \m_{uv}^{v} + \nu_{uv}^{uv} =
	P_{uv}(c^t_u,c^t_v) + P_{uv}(c'_u,c'_v) - P_{uv}(c'_u,c^t_v) - P_{uv}(c^t_u,c'_v) \leq 2 \Delta_{P_{uv}} \ .
	\]
	We thus obtain
	\[
	E[P^{t+1}-P^t\mid F_c^t] \le
	-\sum_{u\in U^t} p_u^t
	+ \sum_{uv\in C_2} 2p_u^tp_v^t \Delta_{P_{uv}}
	= \sum_{u\in U^t} p_u^t  (-1 + \sum_{v|uv\in C_2} p_v^t \Delta_{P_{uv}})  \enspace .
	\]
	Since $\sum_{v|uv\in C_2} \Delta_{P_{uv}}\le \Delta_P$,  $p_u^t\ge \frac{p}{\Delta_P}$, and $p_v^t \le \frac{q}{\Delta_P}$, we obtain
	$
	E[P^{t+1}-P^t\mid F_c^t] \le -p(1-q)\frac{|U^t|}{\Delta_P}  .
	$
	This completes the proof of the lemma.
\qed\end{proof}
\fi 

Theorem~\ref{th:upper_gen} then follows from Theorem~\ref{th:upper_any_game} as we observe that  the global potential is bounded from above by
$
\sum_{uv\in E(G)} \Delta_{P_{uv}} = \sum_u \sum_{v| uv\in E(G)} \Delta_{P_{uv}}/2 \le n \Delta_P/2 .
$
\fi

\ifshort
\smallskip
Since local interaction potential games include \emph{max-cut games}, which are notoriously PLS-complete \cite{SchYan91},  one cannot hope to have convergence  time polynomial independent of `$\Delta_P$' in general.  Local interaction games also include finite \emph{opinion games} \cite{FerGolVen12} and, in particular, 
$
16\Delta \leq \Delta_P \leq 16(\Delta+1)
$, where $\Delta$ is the maximum degree of the underlying graph (details in \cite{full}).  Theorem~\ref{th:upper_gen}  implies:
\else
\subsection{Two Examples of Local Interaction Potential Games}
In this section we discuss two examples of local interaction potential games and how the result of Theorem~\ref{th:upper_gen} relates to them.
\subsubsection{Finite Opinion Games \cite{FerGolVen12}}\label{app:opinion-games}
Consider a game where each player must decide between two opinions, $0$ or $1$, and each player has some internal belief $b_u\in (0,1)$. Since there are only two strategies, $c_u \in \{0,1\}$, better-response and best-response coincide. According to \cite{FerGolVen12}, the payoff of player $u$ depends on the opinion of her neighbors and on her own belief, 
\[
PAY_u(c) = -\left((c_u - b_u)^2 + \sum_{v \in N(u)} (c_u - c_v)^2 \right)\enspace .
\]
and game is a potential game, with the potential function being
\[
P(c) = C(c) + \sum_u (c_u - b_u)^2,
\]
where $C(c)$ is the number of conflicting edges, that is, the edges whose endpoints have different opinions.  We can easily see that these are local interaction potential games by viewing the game as follows:
\begin{itemize}
	\item Between $u$ and every neighbor $v \in N(u)$, we play the symmetric coordination game (payoff is $1$ if they have the same opinion, and $0$ otherwise). 
	\item Each player $u$ is connected to her own \emph{opinion node}, that is, a dummy node $u^*$ who has only strategy $b_u$ available. Between $u$ and $u^*$ we play a ``trivial'' coordination game in which $u$ has payoff $P_{uu^*}(c_u)=-(c_u-b_u)^2$.
\end{itemize}
Finally, in better- or best-response dynamics, one can restrict to $b_u \in \{1/4,3/4\}$ \cite{FerGolVen12}, which makes it possible to have integral potential (simply multiply all payoffs by $16$).
Overall, we can bound the quantity $\Delta_P$ in \eqref{eq:max-potential} as
\[
16\Delta \leq \Delta_P \leq 16(\Delta+1)
\]
where $\Delta$ is the maximum degree in the network (without the dummy nodes $u^*$). 
Theorem~\ref{th:upper_gen} then yields the following:
\fi

\begin{corollary}
	In finite opinion games on networks of maximum degree $\Delta$, the expected converge time of independent better-response dynamics is $O(n \Delta^2)$ whenever $p_u^t = \frac{\alpha}{\Delta}$ for some $\alpha\in [\frac{p}{16}, \frac{q}{16}]$ 
	with $p,q\in(0,1/2)$.
\end{corollary}

\ifshort
\else
\fi

The next example says that one cannot hope to get polynomial convergence time independent of `$\Delta_P$' in general.  

\subsubsection{Max-Cut Games \cite{SchYan91}}\label{ex:max-cut}
Consider a complete \emph{weighted} graph in which every edge $e=(u,v)$ corresponds to a symmetric coordination game rescaled by the weight $w_e$ of this edge,
\[
\begin{game}{2}{2}[][][]
& $B$ & $W$ \\
$B$ & $w_e,w_e$    & $0,0$\\
$W$ & $0,0$    & $w_e,w_e$
\end{game}
\] 
The two possible strategies of each player ($B$ or $W$) determine a partition of the players into two groups, and the potential is simply the value of the cut of the resulting partition (sum of edge weights between the two groups).
These problems are PLS-complete \cite{SchYan91} meaning that one does not expect to have an algorithm that computes a Nash equilibrium in time polynomial in the number of players, regardless of the weights.

Max-cut games are hard when the weights are arbitrary, that is, when the $\Delta_P$ term is not polynomially-bounded. In such instances, even \emph{centralized algorithms} are not expected to find a Nash equilibrium in time polynomial in the number of players.

\section{Conclusion}
This work provides bounds on the time to converge to a (pure Nash) equilibrium when players are active \emph{independently} with some probability and they better or best respond to each others current strategies. Our study focuses on a natural (sub)class of potential games, namely, local interaction potential games. 
\ifshort
The bounds suggest that the time to converge to an equilibrium \emph{must} depend on the \emph{degree} of the nodes in the underlying network 
(cf. Thm.~\ref{th:upper_simple}, Thm.~\ref{th:upper_gen} and Cor.~\ref{cor:lower-bound}).
\else
The bounds suggest that the time to converge to an equilibrium \emph{must} depend on the \emph{degree} of the nodes in the underlying network:
\begin{itemize}
	\item These (distributed) dynamics can \emph{converge quickly} if players are \emph{lazy}, that is, if the probability of being active is inversely proportional to the maximum degree (the precise bounds are given by Theorems~\ref{th:upper_simple} and \ref{th:upper_gen}). 
	\item Conversely, \emph{non-lazy} dynamics can take \emph{exponential} time even in simple cases, where non-lazy means that the probability of being active is too high with respect to the network degree (e.g., when this probability is $1/\Delta^\alpha$ for some $\alpha<1/3$ and $\Delta$ being the maximum degree -- see Corollary~\ref{cor:lower-bound}). 
\end{itemize}
This suggests a sort of threshold effect due to the maximum degree of the network. This parameter has a natural interpretation as it corresponds to the number of players (nodes) that can affect the payoff of a single player (node).  
\fi

\ifshort
Since our bounds hold for local interaction potential games, it would be interesting to investigate whether analogous results hold for \emph{general} potential games. 
Here a relevant notion is that of graphical games \cite{Kea01} and the results in \cite{BabTam14}. 
It would  also be interesting to sharpen some of our bounds to show that $p \simeq 1/\Delta$ is essentially the threshold between fast and slow convergence, and to investigate the range $p\in [1/n, 1/n^{1/3}]$ (cf. Thm~\ref{th:upper_simple} and Cor~\ref{cor:lower-bound}).
\else
Since our bounds hold for local interaction potential games, it would be interesting to investigate whether analogous results hold for \emph{general} potential games
(note that one can always construct a graph representing the dependencies between players, by connecting two players whenever payoff of one depends also on the strategy of the other). 
Here a relevant notion is that of graphical games \cite{Kea01} and the results in \cite{BabTam14}. 
It would  also be interesting to sharpen some of our bounds to show that $p \simeq 1/\Delta$ is essentially the threshold between fast and slow convergence: Is it the case that, for every $\alpha<1$, if all nodes are active with some probability $p \simeq 1/\Delta^{\alpha}$, then the dynamics does not converge in polynomial time in some graphs of maximum degree $\Delta$?
Finally, it would be interesting to  investigate the range $p\in [1/n, 1/n^{1/3}]$ (according to  Theorem~\ref{th:upper_simple} and Corollary~\ref{cor:lower-bound}, this is where convergence time seems to change from polynomial to exponential).
\fi

\smallskip
\noindent
\paragraph{Acknowledgments}
\small
We thank Damien Regnault and Nicolas Schabanel for inspiring discussions on closely related problems, and an anonymous reviewer for pointing out the last open question. Part of this work has been done while the first author was at LIAFA,  Universit\'e Paris Diderot, supported by the French ANR Project DISPLEXITY.

\ifshort
\bibliographystyle{abbrv}
\else 
\bibliographystyle{plain}
\fi 

\bibliography{lazy_graph_coord_CIAC_cameraready}

\ifshort
\end{document}
\else
\fi

\newpage
\appendix

\section{Postponed proofs}

\subsection{Proof of Theorem~\ref{th:upper_any_game}}\label{sec:proof:th_upper_any_game}

We make use of the following lemma which derives 
from classical martingale theory (we include a proof for the sake of completeness).

\begin{lemma}
	\label{lem:martingale}
	Let $(X^t)_{t\in\N}$ a
	sequence of discrete random variables with values in $\{0,\ldots, M\}$,
	and let $T = \min\set{t : X^t = 0}$ be the stopping time
	defined as the random
	variable for the first time $t$ where $X^t= 0$. 
	Suppose that for some $\epsilon >0$  
	\[E[X^{t+1} - X^t \mid X^0=x^0 \wedge \ldots \wedge X^t=x^t]\le -\eps\] 
	for all tuples of values $x^0,\ldots,x^t \in \{1,\ldots, M\}$. Then
	we have $E[T] \le \frac{E[X^0]}{\eps}$.
\end{lemma}

\begin{proof}
	We first show $E[T]<\infty$. As the variables $X^t$ are bounded we can obtain
	a lower bound on the probability that $X^t$ decreases by $\eps/2$ at time $t$:
	for any tuple of positive values $(x^0,\ldots,x^t)$, we consider the event
	$F=\set{X^0=x^0 \wedge \ldots \wedge X^t=x^t}$ and prove
	\begin{equation}
	\label{eq:prob_decr}
	P\bra{X^{t+1} - X^t < -\eps/2 \mid F } \ge \frac{\eps}{2M} \enspace .
	\end{equation}
	From the definition of conditional expectation, we have
	\[
	E[X^{t+1} - X^t \mid F ] = \frac{\int_F (X^{t+1} - X^t) dP}{P[F]} \enspace .
	\]
	The assumption on the expected variation of $X^t$ at time $t$ then implies
	\[
	-\eps P[F] \ge \int_{\set{X^{t+1} - X^t < -\eps/2}\cap F} (X^{t+1} - X^t) dP
	+ \int_{\set{X^{t+1} - X^t \ge -\eps/2}\cap F} (X^{t+1} - X^t) dP  \enspace .
	\]
	Using $X^{t+1} - X^t \ge -M$, we get
	\[
	-\eps P[F] \ge -M P[\set{X^{t+1} - X^t < -\eps/2}\cap F] 
	-\frac{\eps}{2} P[\set{X^{t+1} - X^t \ge -\eps/2}\cap F] \enspace .
	\]
	Equation~\ref{eq:prob_decr}
	is then deduced from $P[F] \ge P[\set{X^{t+1} - X^t \ge -\eps/2}\cap F]$.
	
	Now, for integral $\tau>1$, we have $X^{t+\tau} \le X^t -\eps\tau/2$
	when $X^{s+1}-X^s < -\eps/2$ for $s\in\set{t,\ldots,t+\tau-1}$. We thus have:
	\[
	P[X^{t+\tau} \le X^t -\eps\tau/2 \mid T\ge t+\tau-1] \ge (\eps/2M)^\tau \enspace .
	\]
	Fix $\tau > 2M/\eps$. As $X^t\le M$, we have $X^t-\eps\tau/2<0$ and obtain:
	\[
	P[T\ge t+\tau \mid T\ge t] 
	= 1 - P[X^{t+\tau}=0 \wedge T\ge t+\tau-1\mid T\ge t] 
	< 1 - (\eps/2M)^\tau  \enspace .
	\]
	For integral $k$, we thus have $P[T\ge k\tau] < \alpha^k$
	where $\alpha=1-(\eps/2M)^\tau$.
	Using $P[T=t]\le P[T\ge k\tau]$ for $k\tau\ge t$, we can then write 
	$E[T]\le \sum_{k\ge 1}k\tau\alpha^k\le \alpha\tau/(1-\alpha)^2<\infty$.
	
	\medskip
	
	We finally apply Doob's optional stopping theorem to the random variables
	$Y^t$ defined as follows. Given the outcomes $x^0,\ldots,x^{t-1}$
	of $X^0,\ldots,X^{t-1}$, we set $Y^t=0$ if $x^s=0$ for some $s<t$
	and $Y^t=X^t+\eps t$ otherwise.
	First note that $Y^t$ is a supermartingale:
	Consider values $y^0,\ldots,y^t$. 
	If $y^s=0$ for some $s\le t$, we have $y^{s'}=0$ for $s'\ge s$ and
	$E[Y^{t+1}\mid Y^0=y^0\wedge\ldots\wedge Y^t=y^t]=0=y^t$.
	Otherwise, $y^0,\ldots,y^t$ are all positive and
	$ E[Y^{t+1}\mid Y^0=y^0\wedge\ldots\wedge Y^t=y^t]
	= y^t + E[X^{t+1}-y^t+\eps \mid X^0=y^0\wedge\ldots\wedge X^t=y^t]\le y^t $ 
	according to the assumption on the expected variation of $X^t$ at time $t$.
	Second, we have $E[T]<\infty$ and $|Y^{t+1}-Y^t|\le M+\eps$.
	Doob's theorem thus applies: $Y^T$ is well defined
	and $E[Y^T]\le E[Y^0]$. That is $E[X^T+\eps T]\le E[X^0]$.
	As $X^T= 0$, linearity of expectation yields $E[T] \le E[X^0]/\eps$.
\qed\end{proof}

We are now in a position to prove Theorem~\ref{th:upper_any_game}.

\begin{proof}[of Theorem~\ref{th:upper_any_game}]
	Set $X^t=P^t$ if there are unstable nodes and $X^t=0$ otherwise
        where $P^t$ denotes the value of the potential of the game at time $t$.
	For a given sequence of positive integers $x^0,\ldots,x^t$, consider the
	event $F\subset\Omega$ where all configuration sequences in $F$ 
	satisfy $X^0=x^0,\ldots,X^t=x^t$. 
	For a given
	 configuration $c$ that can occur at time $t$ in $F$, 
         recall that $F_c^t$ denote the set of sequences 
	$c^0,c^1,\ldots$ in $F$ where $c^t=c$. Note that 
	\begin{align*}
		E[X^{t+1} - X^t \mid X^0=x^0 \wedge \ldots \wedge X^t=x^t] &= E[P^{t+1}-P^t | F] \enspace .
	\end{align*}
	Since $F=\cup_c F_c^t$ is a 
	disjoint union we have 
	\begin{align*}		 
		E[P^{t+1}-P^t | F] &= 	\frac{\sum_c E[P^{t+1}-P^t | F_c^t] P[F_c^t]}{P[F]}
		\intertext{since no considered $c$ is an equilibrium and the dynamics is $\eps$-improving}
		E[P^{t+1}-P^t | F] & \leq \frac{\sum_c -\eps P[F_c^t]}{P[F]} = -\eps  \enspace .
	\end{align*}
	By Lemma~\ref{lem:martingale} we get $E[T] \le \frac{M_0}{\eps}$ where $T$ is the expected convergence time of the dynamics (by definition of random variables $X^t$) and $M_0=E[P^0]$ is the expected potential at time 0.
\qed\end{proof}

\subsection{Proof of Lemma~\ref{le:cycle}}\label{app:proof:le:cycle}
\begin{proof}
	By definition of $CYCLE(t)$, at at time $t$ a fraction $\alpha'=\alpha_L$ of nodes in $L$ has color $c$ and a fraction $\alpha_R>1/2$ of nodes in $R$ has color $\overline{c}$. The best response for nodes in $L$ is thus $\overline{c}$. Therefore at time  $t+1$ the fraction $\alpha''$ of nodes in $L$ having color $\overline{c}$ is
	\[
	\alpha'' = 1 - \alpha' + Y/n
	\]  
	where $Y$ is the number of nodes among the $\alpha' n$ of color $c$ that are activated. The expectation of $Y$ is $\mu=p\alpha' n$ and
	Chernoff bounds imply
	\begin{align*}
	P[Y/n\geq (1-\delta)p\alpha'] &\geq  1 -\exp\left(-\frac{\delta^2}{2}\mu\right), \\
	P[Y/n\leq (1+\delta)p\alpha'] &\geq 1 -  \exp\left(-\frac{\delta^2}{3}\mu\right).
	\end{align*}
	We next show two implications saying that the dynamics keeps oscillating `around' the value $\alpha$ as long as $Y$ is close to its expectation according to $\delta=\frac{\epsilon}{1+\epsilon}$:
	\begin{align*}
	Y/n \geq (1-\delta)p\alpha' \text{ and } \alpha'\leq&(1+\epsilon)\alpha & \Rightarrow& & \alpha'' \geq & 1 -\alpha' + p\alpha'  - \delta p \alpha'\\
	& & & & =& 1-\alpha'(1-p) - \delta p \alpha'\\
	& & & & \geq &1-(1+\epsilon)\alpha(1-p)- \delta p (1+\epsilon)\alpha
	\\
	& & & & = &1-\alpha(1-p)  - \epsilon\alpha(1-p) -\delta p (1+\epsilon)\alpha \\
	\text{(since $1 - \alpha(1-p)=\alpha$)}& & & & = &\alpha (1- \epsilon)+ \epsilon\alpha p-\delta p (1+\epsilon)\alpha \\
	\text{(since $\delta=\frac{\epsilon}{1+\epsilon}$)} & & & & \geq &\alpha (1- \epsilon)  \enspace .
	\end{align*}
	Similarly,
	\begin{align*}
	Y/n \leq (1+\delta)p\alpha' \text{ and } \alpha'\geq&(1-\epsilon)\alpha & \Rightarrow& & \alpha'' \le & 1 -\alpha' + p\alpha'  + \delta p \alpha'\\
	& & & & =&1-\alpha'(1-p) + \delta p \alpha'\\
	\text{(using $\alpha'\leq(1+\epsilon)\alpha$)}
	& & & & \leq &1-(1-\epsilon)\alpha(1-p)+ \delta p (1+\epsilon)\alpha
	\\
	& & & & = &1-\alpha(1-p)  + \epsilon\alpha(1-p) +\delta p (1+\epsilon)\alpha \\
	\text{(since $1 - \alpha(1-p)=\alpha$)}& & & & = &\alpha (1+ \epsilon)- \epsilon\alpha p+\delta p (1+\epsilon)\alpha \\
	\text{(since $\delta=\frac{\epsilon}{1+\epsilon}$)}& & & & \leq &\alpha (1+\epsilon)  \enspace .
	\end{align*}
	A symmetric argument applies to $L$ and $R$ exchanged.  We have thus shown that given $CYCLE(t)$ event $CYCLE(t+1)$ holds unless $Y > (1+\delta)\mu$ or $Y < (1-\delta)\mu$, or a symmetric situation holds with $L$ and $R$ exchanged. Then the union bound implies  $P[\neg CYCLE(t+1)| \ CYCLE(t)]\leq 4\exp\left(-\frac{\delta^2}{3}\mu\right)$, which proves the lemma. 
\qed\end{proof}

\subsection{Proof of Corollary~\ref{cor:lower-bound}}\label{app:proof::lower-bound}

\begin{proof}
	In  Lemma~\ref{le:cycle} we have  \[\delta^2\mu = \left(\frac{p}{3+p}\right)^2 p(1+\frac{p}{3})\alpha n= \frac{p^3n}{(3+p)3(2-p)}>\frac{p^3n}{24}\geq \frac{n^{1-3c}}{24}\]
	and the expected number of steps to converge to an equilibrium is $$\Omega(\exp(\delta^2\mu/3))=\Omega(\exp(n^{1-3c}/72))=e^{\Omega(n^{1-3c})}.$$
	This completes the proof.
\qed\end{proof}

\subsection{Proof of Theorem~\ref{th:adaptive}}\label{app:proof:th:adaptive}

\begin{proof}
	We adapt the proof of Lemma~\ref{lem:upper_simple} and show that these dynamics are $\delta$-improving for $\delta=p(1-2q)$.
	Recall that $C_2$ denotes the conflicting edges between two unstable nodes
	at time $t$. We thus have $d_u^t=\card{\set{v \mid uv\in C_2}}$. 
	From Equation~\ref{eq:deltapot2} in the proof of Lemma~\ref{lem:upper_simple},
	we have 
	\[
	E[C^{t+1}-C^t \mid F_c^t] \le
	\sum_{uv\in C_2} 2p_u^tp_v^t - \sum_{u\in U^t} p_u^t  \enspace .
	\]
	In the first sum, we can 
	associate the weight $2p_u^tp_v^t$ to node $u$ if $p_u^t\ge p_v^t$ and
	node $v$ otherwise. We thus get 
	\[
	E[C^{t+1}-C^t \mid F_c^t] \le
	\sum_{u\in U^t} p_u^t\paren{-1 + \sum_{v\mid uv\in C_2, p_v^t\le p_u^t} 2p_v^t}
	\le \sum_{u\in U^t} p_u^t\paren{-1 + 2p_u^td_u^t} \enspace .
	\]
	Using $p_u^t\in [\frac{p}{d_u^t+1}, \frac{q}{d_u^t+1}]$,
	we obtain
	\[
	E[C^{t+1}-C^t \mid F_c^t] \le
	-p(1-2q)\sum_{u\in U^t} \frac{1}{d_u^t+1}  \enspace .
	\]
	As $\card{U^t} > \max_u d_u^t$, we finally have
	\[
	E[C^{t+1}-C^t \mid F_c^t] \le -p(1-2q)  \enspace .
	\]
	The theorem follows from Theorem~\ref{th:upper_any_game} and from the fact that the potential is the number of conflicting edges. 
\qed\end{proof}

\subsection{Proof of Theorem~\ref{th:degree}}\label{app:proof:th:degree}

\begin{proof}
	Similarly to the proof of Theorem~\ref{th:adaptive}, we obtain
	\[
	E[C^{t+1}-C^t \mid F_c^t] \le
	-p(1-2q)\sum_{u\in U^t} \frac{1}{\delta_u} \le -p(1-2q) \frac{|U^t|}{\Delta} 
	\enspace .
	\]
	We then conclude similarly as in the proof of Theorem~\ref{th:upper_simple}.
\qed\end{proof}

\end{document}